\documentclass[10pt, onecolumn]{IEEEtran}
\usepackage{amsmath,amssymb,euscript,yfonts,psfrag,latexsym,graphicx}
\usepackage{bbm,color,amstext,wasysym,subfig,flushend,parskip,balance}
\graphicspath{{./},{./figures/}}

\newtheorem{thm}{Theorem}

\newtheorem{problem}[thm]{Problem}
\newtheorem{remark}[thm]{Remark}

\newtheorem{proof}{Proof}

\newcommand{\mW}{{\mathbf W}}
\newcommand{\mmW}{{\mathbb W}}

\newcommand{\mE}{{\mathbb E}}
\newcommand{\mH}{{\mathbf H}}
\newcommand{\mmH}{{\mathbb H}}
\newcommand{\mmF}{{\mathbb F}}
\newcommand{\mmJ}{{\mathbb J}}

\newcommand{\mR}{{\mathbb R}}

\newcommand{\cL}{{\mathcal L}}
\newcommand{\cN}{{\mathcal N}}

\newcommand{\tr}{\operatorname{trace}}

\newcommand{\trace}{\operatorname{tr}}

\definecolor{grey}{rgb}{0.6,0.6,0.6}
\definecolor{lightgray}{rgb}{0.97,.99,0.99}

\begin{document}
\title{Stochastic control and non-equilibrium\\
thermodynamics: fundamental limits}


\author{Yongxin Chen, Tryphon Georgiou and Allen Tannenbaum
\thanks{Y.\ Chen is with the Department of Electrical and Computer Engineering,
Iowa State University, Ames, Iowa 50011; email: yongchen@iastate.edu}
\thanks{T. \ Georgiou is with the Department of Mechanical \& Aerospace Engineering,
University of Calfornia, Irvine, CA 92697-3975; email: tryphon@uci.edu}
\thanks{A.\ Tannenbaum is with the Departments of Computer Science and Applied Mathematics \& Statistics,
Stony Brook University, Stony Brook, NY 11794; email: allen.tannenbaum@stonybrook.edu}}

\maketitle

\begin{abstract}
We consider damped stochastic systems in a controlled (time-varying) quadratic potential and study their transition between specified Gibbs-equilibria states in finite time. By the second law of thermodynamics, the minimum amount of work needed to transition from one equilibrium state to another is the difference between the Helmholtz free energy of the two states and can only be achieved by a reversible (infinitely slow) process. The minimal gap between the work needed in a finite-time transition and the work during a reversible one, turns out to equal the square of the optimal mass transport (Wasserstein-2) distance between the two end-point distributions times the inverse of the duration needed for the transition. This result, in fact, relates non-equilibrium optimal control strategies (protocols) to gradient flows of entropy functionals via and the Jordan-Kinderlehrer-Otto scheme. The purpose of this paper is to introduce ideas and results from the emerging field of stochastic thermodynamics in the setting of classical regulator theory, and to draw connections and derive such fundamental relations from a control perspective in a multivariable setting.
\end{abstract}

\section{Introduction}
The quest to quantify the efficiency of the steam engine during industrial revolution of the 19th century precipitated the development of thermodynamics. While its birth predates the atomic hypothesis, its modern day formulation makes mention of ``macroscopic'' systems that consist of  a huge number of ``microscopic'' particles (e.g., of the order of Avogadro's number), effectively modeled using probabilistic tools. Its goal is to describe transitions between admissible end-states of such macroscopic systems and to quantify energy and heat transfer between the systems and the ``heat bath'' that they may be in contact with. In spite of the name suggesting ``dynamics,'' the classical theory relied heavily on the concept of quasi-static transitions, i.e., transitions that are infinitely slow.
More realistic finite-time transitions has been the subject of ``non-equilibrium thermodynamics,'' a discipline that has not reached yet the same level of maturity, but one which is currently experiencing a rapid phase of new developments. Indeed, recent developments have launched a phase referred to as {\em stochastic thermodynamics} and {\em stochastic energetics} \cite{Jar97a,Jar97b,Cro99,carberry2004fluctuations,seifert2005entropy,SchSei07,jarzynski2007comparison,kawai2007dissipation,sekimoto2010stochastic,jarzynski2011equalities,aurell2012refined,Sei12}, that aims to quantify non-equilibrium thermodynamic transitions. The reader is referred to a nice and detailed review article \cite{Sei12} for an overview of this subject. Our goal in this paper is to develop such a framework, focusing on the stochastic control of linear uncertain systems in a quadratic (controlled) potential, in a way that is reminiscent of what is known as {\em covariance control} \cite{HotSke87,CheGeoPav14a,CheGeoPav14b,CheGeoPav17c}, and obtain simple derivation of fundamental bounds on the required control and dissipation in achieving relevant control objectives.

Specifically, we consider transitions of a thermodynamic system, represented by overdamped motion of particles in a (quadratic) potential, from one stationary stochastic state to another over a finite-time window $[0,t_f]$. The system is modeled by the (vector-valued) Ornstein-Uhlenbeck process
\begin{align}
dx(t)=-Q(t)x(t)dt +\sigma dw(t), \;\; x(0)=x_0,\label{eq:OU}
\end{align}
with $x\in\mR^n$ and $w$ a standard ($\mR^n$-vector-valued) Wiener process representing a thermal bath of temperature $T$; the parameter
\[
\sigma=\sqrt{2k_BT}.
\]
Here $k_B$ is the Boltzmann constant \cite{Jar97b},
the Hookean force field $-Q(t)x(t)$ is the gradient of a time-varying quadratic Hamiltonian
\begin{align}\label{eq:H}
\mH_t(x) = \mH(t,x)=\frac12 x^\prime Q(t) x,
\end{align}
and the controlled parameter $Q(t)=Q(t)^\prime$, $t\in[0,t_f]$, is scheduled so as to steer the system from a
specified initial distribution for $x_0$, to a final one for $x_{f}$, over the specified time window. The random variables $x_0,x_{f}$ are taken to be Gaussian with zero mean and covariances $\Sigma_0,\Sigma_f$, respectively. That is, the distributions of the state at the two end points have probability densities are $\rho_0 = \cN(0, \Sigma_0),~\rho_f = \cN(0,\Sigma_f)$, or more explicitly,
\[
\rho_i(x)=\frac{1}{(2\pi)^{n/2}|\Sigma_i|^{1/2}}e^{-\frac12 x'\Sigma_i^{-1}x},\;\;i\in\{0,f\},
\]
and we seek to determine the minimum amount of work needed to effect the transition.

From a controls perspective, our problem amounts to {\em covariance control} of {\em bilinear} systems. Indeed, the dynamics are driven by the product of the control input $Q(t)$ times the state $x(t)$. By adjusting the quadratic potential, it is possible to steer the system from one Gaussian distribution to another in finite time $t_f$. When this is the case, we are interested in the optimal control strategy ($Q(t)$, $t\in[0,t_f]$) that minimizes the required control energy.

As noted in the abstract, this minimum control energy is greater than the Helmholtz free energy difference $\Delta \mmF$ between the two states (second law of thermodynamics). Starting with the works by Jarzynski \cite{Jar97a,Jar97b} and Crooks \cite{Cro99}, great new insights began to shed light on the precise amount of work required for such finite-time transitions. Most famously,
the Jarzynski equality
\begin{equation} \label{Jar1}
e^{-\beta\Delta \mmF} =\mE \{e^{-\beta \mW}\},
\end{equation}
relates the {\em equilibrium quantity} $\Delta \mmF$ (free energy difference between equilibrium states) to an {\em averaged
non-equilibrium quantity} (exponential of the work; see our discussion below) over possible trajectories of the system in any finite-time transition. Throughout, $\mE\{\cdot\}$ denotes the expectation on the path space of system trajectories and
\[
\beta=(k_BT)^{-1},
\]
where again $T$ represents temperature of the heat bath and $k_B$ the Boltzmann constant; $\beta$ has units of ``inverse-work.''
The Jarzynski identity holds for arbitrary time-dependent driving force and not necessarily gradient of a quadratic potential.
This type of  result has led to a number of so-called {\em Fluctuation Theorems} in the literature, some of which have profound implications in biology and medicine \cite{Sei12,England,SanGeoTan15}.

Although the Jarzynski equality is quite remarkable, it doesn't provide an explicit gap between the free energy difference $\Delta\mmF$ and the average work $\mmW=\mE \{\mW\}$. This gap is essential if we would like to find an optimal strategy with minimum work to move a thermodynamical system from one state to another. Following up on the Jarzynski equality, the authors of \cite{SchSei07,Gomez2008} analyze the minimum energy control problems in the cases of a Brownian particle dragged by a harmonic optical trap through a viscous fluid, and of a Brownian particle subject to an optical trap with time dependent stiffness, in both overdamped and underdamped setting. Further, in \cite{AurMejMur11,aurell2012refined}, the authors provide an optimal solution that relates the work dissipation to a Wasserstein distance. It can be viewed as a stronger version of the Second Law of Thermodynamics for certain Langevin stochastic processes in finite-time.

The present work is closely related to both \cite{AurMejMur11,aurell2012refined} as well as \cite{SchSei07,Gomez2008}. Compared to \cite{AurMejMur11,aurell2012refined}, our approach gives a control-theoretic account to the fluctuation type results in the case for Gaussian distributions. In addition, we provide an alternative proof for general cases with connections to the gradient flows
 with respect to the Wasserstein geometry \cite{JorKinOtt98}.
 The major difference to \cite{SchSei07,Gomez2008} is that we consider the general matrix cases in this paper. We remark that the problems studied in \cite{AurMejMur11,aurell2012refined} and \cite{SchSei07,Gomez2008} are not equivalent. These two can be connected through an relaxation step as discussed in Section \ref{sec:relax}.

 The rest of the paper is organized as follows. In Section \ref{sec:pre} we go over some key concepts in stochastic thermodynamics and optimal mass transport. The minimum energy control problem between two zero-mean Gaussian distributions is formulated and solved in Section \ref{sec:regulation}. The results' implication in the second law of thermodynamics is discussed in Section \ref{sec:heat}. The result is extended to the nonzero mean setting in Section \ref{sec:nonzero}. A modification of our problem without terminal constraint on distributions is solved in Section \ref{sec:relax}. After that, in Section \ref{sec:JKO}, by leveraging the optimal mass transport theory, we solve the minimum energy control problem with general marginal distributions. Last, for comparison, we go over a simple proof of the Jarzynski equality in Section \ref{sec:jarzynski}. We conclude with several numerical examples in Section \ref{sec:eg}.

 \section{Preliminaries}\label{sec:pre}
 This work bridges stochastic control, stochastic thermodynamics and optimal mass transport. Below we introduce some key concepts in stochastic thermodynamics and optimal mass transport that are relevant.

 \subsection{Stochastic thermodynamics}
Stochastic thermodynamics \cite{Sei12,Owen84} is one approach to study thermodynamical systems via stochastic calculus. A basic model in this framework is
	\begin{equation}
	dx(t)=-\nabla \mH(t,x(t))dt +\sigma dw(t).
	\end{equation}
 Here $\mH$ is the Hamiltonian of the system and the noise $dw$ describes the effect of the heat bath. When the Hamiltonian is fixed, the state distribution converges to a Boltzmann distribution
	\[
		\rho_B(x) = \frac{1}{Z} e^{-\beta \mH(x)},
	\]
where $Z$ is a partition function. This is known as the equilibrium steady state.
We denote the internal energy and Helmholtz free energy in the equilibrium steady state by $\mmH$ and $\mmF$ respectively. They are defined by	\cite{Owen84}
	\[
		\mmH:= \mmH(\rho_B) := \int \mH(x) \rho_B(x) dx,
	\]
and
	\[
		\mmF :=\mmF(\mH) = -k_B T \log Z.
	\]
Clearly, they satisfy the relation
 	\begin{equation}\label{eq:free}
		\mmF = \mmH - TS(\rho_B)
	\end{equation}
with the entropy being
	\[
		S(\rho) = -k_B \int \rho(x) \log \rho(x) dx.
	\]

 The above relation \eqref{eq:free} may be used to extend the definition of free energy to non-equilibrium states. More precisely, let $\rho$ be the probability distribution of the state, then we can define the free energy through \cite{ParHorSag15}
	\begin{equation}\label{eq:freenon}
		\mmF(\rho;\mH) = \mmH(\rho) - TS(\rho).
	\end{equation}
Note that
	\[
		\mmF(\rho;\mH) \ge \mmF(\rho_B;\mH) = \mmF.
	\]

 \subsection{Optimal mass transport}
We only cover concepts that are related to the present work. We refer the reader to \cite{Vil03} for complete details.
Consider two measures $\rho_0, \rho_1$ on ${\mathbb R}^n$ with equal total mass. Without loss of generality, we take $\rho_0$ and $\rho_1$ to be probability distributions.
In the Kantorovich's formulation of optimal mass transport with quadratic cost, one seeks a joint distribution $\pi\in\Pi(\rho_0,\rho_1)$ on $\mR^n\times\mR^n$, referred to as ``coupling" of $\rho_0$ and $\rho_1$, that minimizes the total cost, and so that the marginals along the two coordinate directions coincide with $\rho_0$ and $\rho_1$, respectively, that is,
    \begin{equation}\label{eq:OptTrans}
        \inf_{\pi\in\Pi(\rho_0,\rho_1)}\int_{\mR^n\times\mR^n}\|x-y\|^2\pi(dxdy).
    \end{equation}

The above optimal transport problem has a surprising stochastic control formulation, which reads as
	\begin{subequations}\label{eq:stochcontrol}
	\begin{eqnarray}
	&& \inf_{u}~ \mE \left\{\int_0^1 \|u(t,x(t))\|^2 dt\right\}
	\\&& \dot x(t) = u(t,x(t))
	\\&& x(0) \sim \rho_0, ~~x(1) \sim \rho_1.
	\end{eqnarray}
	\end{subequations}
Briefly, we seek a feedback control strategy with minimum energy that drives the state of an integrator from an initial probability distribution $\rho_0$ to a terminal probability distribution $\rho_1$.
	
Both of the above problems have unique solutions under the assumption that the marginal distributions are absolutely continuous. The square root of the minimum of the cost (\eqref{eq:OptTrans} or \eqref{eq:stochcontrol}) defines a Riemannian metric on $P_2(\mR^n)$, the space of probability distributions on $\mR^n$ with finite second-order moments. This metric is known as the \emph{Wasserstein metric} $W_2$ \cite{JorKinOtt98,Otto,Vil03,Vil08}. On this Riemannian-type manifold, the geodesic curve connecting $\rho_0$ and $\rho_1$ is given by $\rho_t$, the probability density of $x(t)$ under the optimal control policy.
This is called \emph{displacement interpolation} \cite{McC97} and it satisfies
	\begin{equation}\label{eq:W2geodesic}
		W_2(\rho_s,\rho_t) = (t-s) W_2(\rho_0,\rho_1),\quad 0\le s< t\le 1.
	\end{equation}
	
When both of the marginals $\rho_0, \rho_1$ are Gaussian distributions, the problem has a closed-form solution \cite{Tak11,dowson1982frechet,jiang2012geometric}. Denote the mean and covariance of $\rho_i, i=0,1$ by $m_i$ and $\Sigma_i$, respectively. Let $X, Y$ be two Gaussian random vectors associated with $\rho_0, \rho_1$, respectively. Then the cost in (\ref{eq:OptTrans}) becomes
	\begin{equation}\label{eq:expectcost}
		\mE\{\|X-Y\|^2\} = \mE \{\|\tilde X-\tilde Y\|^2\} +\|m_0-m_1\|^2,
	\end{equation}
where $\tilde X = X-m_0, \tilde Y = Y-m_1$ are zero-mean versions of $X$ and $Y$. We minimize \eqref{eq:expectcost} over all the possible Gaussian joint distributions between $X$ and $Y$, which gives
	\begin{equation}\label{eq:OMTSDP}
	\min_S \left\{\|m_0-m_1\|^2+\tr(\Sigma_0+\Sigma_1-2S) ~\mid~
	\left[\begin{matrix}
	\Sigma_0 & S \\ S' & \Sigma_1
	\end{matrix}\right]\ge 0\right\},
	\end{equation}
with $S=\mE \{\tilde X \tilde Y'\}$. The constraint is a semidefinite one, so the above problem is one of semidefinite programming (SDP). The minimum is achieved in closed-form by the unique minimizer
	\begin{equation}\label{eq:S}
		S= \Sigma_0^{1/2}(\Sigma_0^{1/2}\Sigma_1\Sigma_0^{1/2})^{1/2}\Sigma_0^{-1/2}
	\end{equation}
corresponding to the minimum value
	\begin{equation}\label{eq:W2gaussian}
		W_2(\rho_0,\rho_1)^2 = \|m_0-m_1\|^2+\tr(\Sigma_0+
		\Sigma_1-2(\Sigma_0^{1/2}\Sigma_1\Sigma_0^{1/2})^{1/2}).
	\end{equation}
The resulting displacement interpolation $\rho_t$ is a Gaussian distribution with mean $m_t = (1-t)m_0+tm_1$ and covariance
	\begin{equation}\label{eq:disinterpG}
		\Sigma_t = \Sigma_0^{-1/2} \left((1-t)\Sigma_0+t(\Sigma_0^{1/2}\Sigma_1\Sigma_0^{1/2})^{1/2}
		 \right)^2 \Sigma_0^{-1/2}.
	\end{equation}

\section{Regulation via a time-varying potential} \label{sec:regulation}

We consider the stochastic dynamical system in \eqref{eq:OU}.
As mentioned earlier, it represents a thermodynamical system with a quadratic Hamiltonian \eqref{eq:H}, overdamped and attached to a heat bath that is modeled by the stochastic excitation $dw$. The initial state is a Gaussian random vector $x_0\sim {\mathcal N}(0,\Sigma_0)$, i.e., one having covariance $\Sigma_0$ and mean $\mE \{x_0\}=0$.
The initial distribution is usually taken to be the stationary distribution with potential remaining constant on $(-\infty, 0]$ by keeping $Q(t)\equiv Q_0$ over $t\in(-\infty, 0]$, in which case $Q_0=\frac{\sigma^2}{2}\Sigma_0^{-1}$, but this assumption is not required.
We are interested in steering the state to the terminal distribution
${\mathcal N}(0,\Sigma_f)$ through selecting an optimal (least energy) time-varying control matrix variable $Q(\cdot)=Q'(\cdot)$ satisfying the boundary conditions $Q(0) = Q_0, Q(t_f) = Q_f$.

The control energy (work) delivered to the system along any particular sample path $x(\cdot)$ by the time-varying potential \eqref{eq:H} is
\[
\mW(Q,x):=\int_0^{t_f}\frac{\partial \mH(t,x)}{\partial t} dt = \int_0^{t_f}\langle \dot Q(t), \frac{\partial \mH(t,x)}{\partial Q}\rangle dt,
\]
where $\langle X,Y\rangle = \trace(X'Y)$.
Thus, by averaging over all possible sample paths, we obtain
	\begin{eqnarray}\nonumber
	\mmW &:=& \mE\{\mW(Q,x)\} = \mE \left\{\int_0^{t_f}\langle \dot Q, \frac{\partial \mH}{\partial Q}\rangle dt\right\}\\\nonumber
	&=&  \mE \left\{\int_0^{t_f} \frac{1}{2} \langle \dot Q(t), x(t)x(t)'\rangle dt\right\}\\\nonumber
	&=&  \frac{1}{2}\int_0^{t_f} \langle \dot Q(t), \Sigma(t)\rangle dt.
	\end{eqnarray}
Here, $\Sigma(\cdot)$ is the state covariance which, according to standard linear systems theory,
evolves according to the Lyapunov equation
	\begin{equation}\label{eq:lyap}
	\dot \Sigma(t) = -Q(t) \Sigma(t) - \Sigma(t) Q(t) + \sigma^2 I.
	\end{equation}
The control may be discontinuous, reflecting instantaneous changes in the Hamiltonian $\mH$, in which case, the expression for the work becomes the Lebesgue-Stieltjes integral
 	\begin{equation}\label{eq:workdisc}
	\mmW 
	 = \frac{1}{2}\int_{0^-}^{t_f^+} \langle dQ(t), \Sigma(t)\rangle,
	\end{equation}
where $0^-,t_f^+$ represent limits from below and above, respectively, so as to account for the discontinuities.


\begin{problem}\label{pro:cov}
Determine a control law
\[
\{Q(t)\mid t\in[0,t_f]\}
\]
that minimizes \eqref{eq:workdisc} subject to \eqref{eq:lyap} and the boundary conditions $Q(0) = Q_0, Q(t_f) = Q_f, \Sigma(0)=\Sigma_0, \Sigma(t_f) = \Sigma_f$.
\end{problem}

\begin{thm}\label{thm:opt}
Problem \ref{pro:cov} has a unique minimizer $Q_{\rm opt}(\cdot)$
as follows:\\[-.3in]
\begin{itemize}
\item[(i)] If $\Sigma_0 = \Sigma_f$, then $\mmW_{\rm min}=0$ and
\begin{align*}
Q_{\rm opt}(t)&=\frac{\sigma^2}{2}\Sigma_0^{-1},\\
\Sigma(t)&=\Sigma_0, \mbox{ for all }t\in(0,t_f).
\end{align*}
\item[(ii)] If $\Sigma_0 \neq \Sigma_f$, then
\begin{equation}
\mmW_{\rm min}=-\frac{\sigma^2}{4}\tr\log(\Sigma_f\Sigma_0^{-1})+\frac{1}{t_f} \tr(\Sigma_0+\Sigma_f-2(\Sigma_0^{1/2}\Sigma_f\Sigma_0^{1/2})^{1/2})
\end{equation}
and
\begin{subequations}
\begin{align}\label{eq:Q}
Q_{\rm opt}(t)&=\frac{\sigma^2}{2}\Sigma(t)^{-1}-(\Lambda(0)^{-1}+tI)^{-1}\\
\Sigma(t)&=(\Lambda(0)^{-1}+tI)M^{-1}(\Lambda(0)^{-1}+tI),
\label{eq:optSigma}
\end{align}
\end{subequations}
	with
\begin{subequations}\label{eq:parameters}
\begin{align}
\Lambda(0)&=\frac{1}{t_f} (-I + \Sigma_0^{-1/2}(\Sigma_0^{1/2}\Sigma_f\Sigma_0^{1/2})^{1/2}\Sigma_0^{-1/2})\\
M&=\Lambda(0)^{-1}\Sigma_0^{-1}\Lambda(0)^{-1}.
\end{align}
\end{subequations}
\end{itemize}
\end{thm}

\begin{proof}
Case (i) is trivial. We only discuss case (ii) in detail. Applying integration by parts to \eqref{eq:workdisc}, we obtain
\begin{equation}\label{eq:finalwork}
\mmW
=-\frac{1}{2}\int_{0^-}^{t_f^+} \langle Q(t), d\Sigma(t)\rangle + \frac12\tr \left(Q(t_f^+)\Sigma_f-Q(0^-)\Sigma_0\right).
\end{equation}
Notice that
	\[
	\frac12\tr \left(Q(t_f^+)\Sigma_f-Q(0^-)\Sigma_0\right)=\frac12\tr \left(Q_f\Sigma_f-Q_0\Sigma_0\right) =\mmH_f(\rho_f)-\mmH_0(\rho_0)
	\]
 is precisely the change in the average energy (expectation of the Hamiltonian) and is independent of the control $\{Q(t),~t\in[0,t_f]\}$.
More specifically,
	\[
		\mmH_f(\rho_f)-\mmH_0(\rho_0)= \int \mH_f(x) \rho_f(x) dx-\int \mH_0(x) \rho_0(x) dx.
	\]
	
Substituting \eqref{eq:lyap} into \eqref{eq:finalwork} yields
\begin{equation}\label{eq:Wnew}
\mmW=\frac12\int_{0}^{t_f} \left(2\tr(Q(t)\Sigma(t)Q(t))-\sigma^2\tr(Q(t))\right)dt +\mmH_f(\rho_f)-\mmH_0(\rho_0).
\end{equation}
We change variables, replacing $Q$ by
\[
\Lambda(t):=\frac{\sigma^2}{2}\Sigma(t)^{-1}-Q(t),
\]
in both, the
constraint \eqref{eq:lyap} as well as \eqref{eq:Wnew}. These now become
\begin{align}\label{eq:diffeq2}
\dot\Sigma(t) &= \Lambda(t)\Sigma(t)+\Sigma(t)\Lambda(t), \mbox{ and}\\
\mmW & = \int_{0}^{t_f} \tr(\Lambda(t)\Sigma(t)\Lambda(t)) dt - \frac{\sigma^2}{2} \int_{0}^{t_f} \tr(\Lambda(t))dt+\mmH_f(\rho_f)-\mmH_0(\rho_0),\label{eq:mmWagain}
\end{align}
respectively. From \eqref{eq:diffeq2},
\[
\tr(\dot \Sigma(t)\Sigma(t)^{-1})=2\tr(\Lambda(t)).
\]
It follows that
\begin{align*}
\frac{\sigma^2}{2} \int_{0}^{t_f} \tr(\Lambda(t))dt &=
\frac{\sigma^2}{4} \int_{0}^{t_f} \tr(\dot \Sigma(t)\Sigma(t)^{-1})dt\\
&=\frac{\sigma^2}{4} \int_{0}^{t_f} \tr(\frac{d}{dt}\log(\Sigma(t)))dt\\
&= \frac{\sigma^2}{4}\tr\log(\Sigma(t_f)\Sigma(0)^{-1})\\
&= \frac{\sigma^2}{4}\tr\log(\Sigma_f\Sigma_0^{-1})
\end{align*}
is independent of the choice of $Q$ or $\Lambda$.
Thus, minimization of \eqref{eq:Wnew} (equivalently, minimization of \eqref{eq:mmWagain}) is equivalent to minimization of
\begin{align}\label{eq:Wnewportion}
\mmJ:=\int_{0}^{t_f} \tr(\Lambda(t)\Sigma(t)\Lambda(t)) dt
\end{align}
subject to the choice of $\Lambda(\cdot)$ that satisfies \eqref{eq:diffeq2} and the boundary conditions $\Sigma(0)=\Sigma_0$ and $\Sigma(t_f)=\Sigma_f$.
Then,
\begin{equation} \label{eq:key}
\mmW=\mmJ-\frac{\sigma^2}{4}\tr\log(\Sigma_f\Sigma_0^{-1})+\mmH_f(\rho_f)-\mmH_0(\rho_0).
\end{equation}

Setting $X:=\Lambda \Sigma$, the functional $\mmJ$ becomes convex in $X,\Sigma$. Then,
\[
\min_\Lambda \mmJ = \min_X \int_{0}^{t_f}\tr(X(t)\Sigma(t)^{-1}X(t)^\prime)dt,
\]
subject to the linear constraint
\begin{align}\label{eq:diffeq3}
\dot\Sigma(t) &= X(t)+X(t)^\prime, \;\Sigma(0)=\Sigma_0,\,\Sigma(t_f)=\Sigma_f,
\end{align}
has a unique solution. In fact, a closed-form expression can be obtained by considering the necessary conditions that are being dictated by the stationarity of the Lagrangian
\begin{align*}
{\mathcal L}(\Sigma,X,\hat\Lambda):=\int_{0}^{t_f}\tr(X(t)\Sigma(t)^{-1}X(t)^\prime)dt\\+
\int_{0}^{t_f}\tr(\hat\Lambda(\dot\Sigma(t)- X(t)-X(t)^\prime))dt.
\end{align*}
Specifically, the first variation with respect to $X$ gives that
\[
\hat\Lambda=X\Sigma^{-1}=\Lambda.
\]
Then, the variation with respect to $\Sigma$ gives
	\begin{equation}
		\dot\Lambda = -\Lambda^2.
	\end{equation}
Assuming that $\Lambda(0)$ is nonsingular,
\[\Lambda(t) = (\Lambda(0)^{-1}+tI)^{-1}.
\]
From \eqref{eq:diffeq2},
	\[
		\Sigma(t) =
		(\Lambda(0)^{-1}+t I)M^{-1}(\Lambda(0)^{-1}+tI)
			\]
for a suitable choice of a matrix $M$. Then, $\Lambda(0),M$ are determined from the boundary conditions,
	\begin{eqnarray*}
		\Sigma(0) &=& \Lambda(0)^{-1}M^{-1}\Lambda(0)^{-1} = \Sigma_0,\\
		\Sigma(t_f) &=& (\Lambda(0)^{-1}+t_f I)M^{-1}(\Lambda(0)^{-1}+t_f I) = \Sigma_f.
	\end{eqnarray*}
It follows that
\[
\Sigma_0^{-1}= (I+\Lambda(0)t_f)\Sigma^{-1}_f(I+\Lambda(0)t_f),
\]
from which we deduce that $I+\Lambda(0)t_f$ is the geometric mean $(\Sigma_0^{-1}\sharp \Sigma_f)$  of $\Sigma_0^{-1}$ and $\Sigma_f$ (see \cite{bhatia2013matrix}),
viz.,
\[
I+\Lambda(0)t_f = \Sigma_0^{-1/2}(\Sigma_0^{1/2}\Sigma_f\Sigma_0^{1/2})^{1/2}\Sigma_0^{-1/2}.
\]
Thus, we conclude \eqref{eq:parameters}.

Finally, plugging the optimal solution into \eqref{eq:Wnewportion} yields
	\begin{eqnarray}
	\mmJ_{\rm min}&=&t_f \tr(M^{-1}) \nonumber
	\\&=& t_f \tr(\Lambda(0)\Sigma_0\Lambda(0)) \nonumber
	\\&=& \frac{1}{t_f} \tr(\Sigma_0+\Sigma_f-2(\Sigma_0^{1/2}\Sigma_f\Sigma_0^{1/2})^{1/2}), \label{eq:Jmin}
	\end{eqnarray}
which completes the proof. $\Box$
\end{proof}
\begin{remark}
The optimal control $Q(t)$ in \eqref{eq:Q} is continuous function on $(0,\,t_f)$. The limit values at $t= 0, t_f$ are
	\[
		Q(0^+) = \frac{\sigma^2}{2}\Sigma_0^{-1}
+\frac{1}{t_f}(I - \Sigma_0^{-1/2}(\Sigma_0^{1/2}\Sigma_f\Sigma_0^{1/2})^{1/2}\Sigma_0^{-1/2})
	\]
and
	 \[
	 	Q(t_f^-) = \frac{\sigma^2}{2}\Sigma_f^{-1}+\frac{1}{t_f}(-I + \Sigma_0^{1/2}(\Sigma_0^{1/2}\Sigma_f\Sigma_0^{1/2})^{-1/2}\Sigma_0^{1/2})
	\]
respectively. These may not be consistent with the boundary conditions $Q(0) = Q_0, Q(t_f)= Q_f$, which dictates the discontinuities of the optimal control at $t= 0, t_f$. When both the initial and terminal states are stationary, namely, $Q_0=\frac{\sigma^2}{2}\Sigma_0^{-1}, Q_f=\frac{\sigma^2}{2}\Sigma_f^{-1}$, such discontinuities go to zero as the length of time $t_f$ goes to infinity.
\end{remark}


\section{Second law of thermodynamics and optimal transport} \label{sec:heat}

The problem to minimize $\mmJ$ in \eqref{eq:Wnewportion} is in fact a Monge-Kantorovich optimal transport problem with marginals $\rho_0$ and $\rho_f$, and
quadratic cost functional \cite{Vil03, Vil08}. Specifically,
\begin{equation} \label{eq:key1}
\min_\Lambda \mmJ = \frac{1}{t_f}W_2(\rho_0,\rho_f)^2.
\end{equation}
This follows directly from \eqref{eq:Jmin} and \eqref{eq:W2gaussian}.
Alternatively, consider the stochastic control formulation \eqref{eq:stochcontrol} of optimal transport.
The optimal solution, see e.g., \cite{CheGeoPav15b,CheGeoPav14e}, is in the linear state feedback form
 $u(t,x)=\Lambda(t)x$. With $\mE\{x(t)x(t)'\}=\Sigma(t)$,
 \[
 \mE\{ \|u\|^2\} = \tr(\Lambda(t)\Sigma(t)\Lambda(t)')
 \]
and
\[
	\dot\Sigma(t) = \Lambda(t)\Sigma(t)+\Sigma(t)\Lambda(t)'.
\]
The optimal $\Lambda$ is symmetric and therefore coincides with the minimizer of $\mmJ$ up to a scaling in time.
%
The factor $1/{t_f}$ shows up due to the fact that the time window in standard optimal transport is $[0,\,1]$ while in our problem it is $[0,\,t_f]$.
Naturally, it follows from this equivalence that the probability density flow of $x(t)$ under optimal control $Q_{\rm opt}$ is a (scaled) geodesic (displacement interpolation) between $\rho_0$ and $\rho_f$ with respect to the Wasserstein metric $W_2$. Indeed, it can be verified that $\Sigma$ in \eqref{eq:optSigma} is
	\[
	\Sigma(t) = \Sigma_0^{-1/2} \left((1-\frac{t}{t_f})\Sigma_0+\frac{t}{t_f}(\Sigma_0^{1/2}\Sigma_f\Sigma_0^{1/2})^{1/2}
		 \right)^2 \Sigma_0^{-1/2},
	\]
which is consistent with the geodesic formula in \eqref{eq:disinterpG}. Thus, we obtain the following:
\begin{thm}
The probability density flow of $x(t)$ in Problem \ref{pro:cov} with optimal control $Q$ is the (scaled) displacement interpolation between $\rho_0$ and $\rho_f$.
\end{thm}

From \eqref{eq:key1} and Theorem \ref{thm:opt}, the minimum of $\mmW$ is
\[
\frac{1}{{t_f}} W_2(\rho_0,\rho_f) ^2 - \frac{\sigma^2}{4}\tr\log(\Sigma_f\Sigma_0^{-1})+\mmH_f(\rho_f)-\mmH_0(\rho_0).
\]
Using the ``log det = trace log'' equality, and the fact that the entropy of Gaussian distributions is
	\[
		S(\rho) = -k_B\int \rho\log\rho= \frac{k_B}{2} \log\det(\Sigma)+\frac{k_B}{2} \log\det(2\pi I)+\frac{k_B}{2} \tr I,
	\]
and, in view of $\sigma^2= 2k_BT$, we get that the minimum value of $\mmW$ is
\begin{equation} \label{eq:min}
\mmW_{\min} =\frac{1}{{t_f}} W_2(\rho_0,\rho_f) ^2 - T S(\rho_f) + T S(\rho_0)+\mmH_f(\rho_f)-\mmH_0(\rho_0) = \frac{1}{{t_f}} W_2(\rho_0,\rho_f) ^2 - T \Delta S +\mmH_f(\rho_f)-\mmH_0(\rho_0).
\end{equation}
Next note that the change in the Helmholtz free energy (see \eqref{eq:freenon}) is
\[
\Delta\mmF=\mmF(\rho_f;\mH_f) - \mmF(\rho_0;\mH_0) =\mmH_f(\rho_f)-\mmH_0(\rho_0)-T\Delta S.
\]
Putting all this together, we get that
\begin{thm}\label{thm:WFnon}
	\begin{equation}\label{eq:WFnon}
	\mmW_{\min} = \Delta \mmF +\frac{1}{t_f} W_2(\rho_0,\rho_f) ^2.
	\end{equation}
\end{thm}

Recall that for reversible processes, one has
	\[
		\mmW = \Delta \mmF,
	\]
and for general processes
	\[
		\mmW \ge \Delta \mmF.
	\]
These are equivalent to the second law of thermodynamics, which says that the total entropy of an isolated system is nondecreasing. Theorem \ref{thm:WFnon} provides a stronger lower bound for entropy production of a finite-time process, and this bound connects thermodynamics and optimal mass transport!
	
The difference $\mmW- \Delta \mmF$ is the entropy production, or {\em work dissipation}, and denoted $\mmW_{\rm diss}$. This is the same as $ \mmJ$ in the proof of Theorem \ref{thm:opt}.
Theorem \ref{thm:WFnon} provides a fundamental lower bound of work dissipation
	\[
		\mmW_{\rm diss}\ge \frac{1}{t_f} W_2(\rho_0,\rho_f) ^2
	\]
for a irreversible process evolving in a finite time-interval $[0,\,t_f]$. As we discussed earlier, this lower bound is achieved by the optimal protocol \eqref{eq:Q} and the corresponding  probability density flow is the displacement interpolation between $\rho_0$ and $\rho_f$.

In general, for any feasible protocol $Q(t), t\in[0,\,t_f]$, the work dissipation depends only on the probability density flow $\rho_t$ from $\rho_0$ to $\rho_f$.
\begin{thm}\label{thm:diss}
	\begin{subequations}\label{eq:worksubopt}
	\begin{eqnarray}
	\mmW_{\rm diss} &=& \int_0^{t_f} \tr(\Lambda(t)\Sigma(t)\Lambda(t)) dt \label{eq:worksubopt1}
	\\&&
	\dot \Sigma(t) = -Q(t) \Sigma(t) - \Sigma(t) Q(t) + \sigma^2 I \label{eq:worksubopt2}
	\\&& \Lambda(t)=\frac{\sigma^2}{2}\Sigma(t)^{-1}-Q(t). \label{eq:worksubopt3}
	\end{eqnarray}
	\end{subequations}
\end{thm}

Indeed, once the probability density flow $\rho_t$ is fixed, we can get $Q, \Lambda$ through \eqref{eq:worksubopt2}-\eqref{eq:worksubopt3} and then $\mmW_{\rm diss}$ through \eqref{eq:worksubopt1}. In fact, this is nothing but the length (scaled  by $t_f$) of the curve $\rho_t$ on the manifold of probability densities equipped with the Wasserstein metric $W_2$ \cite{Vil08}.

\begin{remark}
Minimizing the work $\mmW$ is equivalent to minimizing the work dissipation $\mmW_{\rm diss}=\mmW-\Delta \mmF$ as $\Delta\mmF$ relies only on the boundary conditions. When there is no constraint on the choice of Hamiltonian $\mH$, the optimal strategy is given by Theorem \ref{thm:opt}, which leads to a probability density flow that is the displacement interpolation between $\rho_0$ and $\rho_f$. On the other hand, when there exist constraints on $\mH$, in view of the above argument, we can lift the problem to the space of probability densities, and seek a feasible time-varying Hamiltonian such that the resulting density flow $\rho_t$ has minimum length on the manifold of probability densities equipped with the Wasserstein metric $W_2$. This may lead to a promising direction to solve constrained thermodynamical control problems.
\end{remark}

\section{Hamiltonian with nonzero center}\label{sec:nonzero}
In this section, we extend our framework to the cases when the centers of the Hamiltonian potentials are allowed to change over time. Specifically, consider the stochastic thermodynamical system
	\begin{equation}\label{eq:dyndrift}
		dx(t)=-Q(t)(x(t)-p(t))dt +\sigma dw(t),
	\end{equation}
which corresponds to the Hamiltonian
	\[
		\mH(t,x)=\frac12 (x-p(t))^\prime Q(t) (x-p(t))
	\]
with time-varying center $p(t)$.
Assume the initial and terminal Gaussian distributions are $\rho_0=\cN(m_0,\Sigma_0), \rho_f = \cN (m_f,\Sigma_f)$. Our goal is to drive the system from initial distribution $\rho_0$ to terminal distribution $\rho_f$ with minimum cost via changing the strength $Q$ as well as the center $p$ of the potential well at the same time.

The mean and covariance of $x(t)$ evolve according to
	\begin{subequations}\label{eq:Qp}
	\begin{equation}\label{eq:Qp1}
		\dot\Sigma(t) =-Q(t)\Sigma(t)-\Sigma(t)Q(t)+\sigma^2 I
	\end{equation}
and
	\begin{equation}\label{eq:Qp2}
		\dot m(t) = -Q(t) m(t)+Q(t)p(t).
	\end{equation}
	\end{subequations}
The average work is
	\begin{eqnarray*}
	\mmW &=& \mE \left\{\int_{0^-}^{t_f^+}\frac{\partial \mH(t,x)}{\partial t}\right\}\\
	&=&  \mE \left\{\int_{0^-}^{t_f^+} \frac{1}{2} \langle \dot Q(t), (x(t)-p(t))(x(t)-p(t))'\rangle -(x(t)-p(t))'Q(t)\dot p(t)dt\right\}\\
	&=&  \int_{0^-}^{t_f^+} [\frac{1}{2}\langle \dot Q(t), \Sigma(t)\rangle+\frac12 (m(t)-p(t))'\dot{Q}(t)(m(t)-p(t))- (m(t)-p(t))'Q(t)\dot p(t)]dt
	\\&=& \int_{0^-}^{t_f^+}  [\frac{1}{2}\langle dQ(t), \Sigma(t)\rangle+d(\frac{1}{2} (m(t)-p(t))'Q(t)(m(t)-p(t)))+\|\dot m(t)\|^2dt].
	\end{eqnarray*}
\begin{problem}
Find a time-varying $Q(t)$ from $Q(0) = Q_0$ to $Q(t_f) = Q_f$ and a time-varying $p(t)$ from $p(0)=p_0$ to $p(t_f)=p_f$ such that $x(t)$ has $\rho_0, \rho_f$ as the marginal distributions and the average work is minimized.
\end{problem} 	

\begin{thm}\label{thm:optQp}
The optimal $Q$ and $\Sigma$ are as in the zero-mean case, and the optimal $m, p$ satisfy
	\begin{equation}
		m(t) = \frac{t_f-t}{t_f}m_0+\frac{t}{t_f}m_f ,
	\end{equation}
and
	\begin{equation}\label{eq:optp}
		p(t) = \frac{t_f-t}{t_f}m_0+\frac{t}{t_f}m_f +\frac{1}{t_f}Q(t)^{-1}(m_f-m_0).
	\end{equation}
The corresponding work is
	\begin{eqnarray}\nonumber
		\mmW_{\rm min} &=& -\frac{\sigma^2}{4}\tr\log(\Sigma_f\Sigma_0^{-1})+\frac{1}{t_f} \tr(\Sigma_0
		+\Sigma_f-2(\Sigma_0^{1/2}\Sigma_f\Sigma_0^{1/2})^{1/2})+\frac{1}{t_f}\|m_f-m_0\|^2
		\\&&+\frac{1}{2} (m_f-p_f)'Q_f(m_f-p_f)-\frac{1}{2} (m_0-p_0)'Q_0(m_0-p_0).\label{eq:workminnon}
	\end{eqnarray}
\end{thm}
\begin{proof}
We first simplify the work to
	\begin{eqnarray*}
	\mmW &=& \int_{0^-}^{t_f^+}  [\frac{1}{2}\langle dQ(t), \Sigma(t)\rangle+d(\frac{1}{2} (m(t)-p(t))'Q(t)(m(t)-p(t)))+\|\dot m(t)\|^2dt]
	\\&=& \int_{0^-}^{t_f^+} \frac{1}{2}\langle dQ(t), \Sigma(t)\rangle+\int_{0}^{t_f} \|\dot m(t)\|^2dt+\frac{1}{2} (m_f-p_f)'Q_f(m_f-p_f)
	-\frac{1}{2} (m_0-p_0)'Q_0(m_0-p_0).
	\end{eqnarray*}
The last two terms depend only on the boundary conditions. The first two terms are totally decoupled; one depends only on $Q, \Sigma$ while the other on $m$. Therefore, we can minimize these two terms independently. Clearly, the optimal $Q,\Sigma$ are identical to that in the zero-mean case (Theorem \ref{thm:opt}). To obtain $m, p$, we minimize $\int_0^{t_f}\|\dot m(t)\|^2dt$ subject to the boundary conditions $m(0)=m_0, m(t_f)=m_f$. Thus, the optimal $m$ is the linear interpolation between $m_0$ and $m_f$. Plugging it into \eqref{eq:Qp2} concludes the optimal $p$. $\Box$
\end{proof}

As we have already seen in Section \ref{sec:regulation}, the optimal strength $Q$ of the potential usually has discontinuities at the boundary points $t=0,\, t_f$. We next argue that similar phenomenon happens for the center $p$ of the potential. From \eqref{eq:optp} we get
	\begin{eqnarray*}
		p(0^+) &=& m_0+\frac{1}{t_f}Q(0^+)^{-1}(m_f-m_0),
		\\
		p(t_f^-) &=& m_f +\frac{1}{t_f}Q(t_f^-)^{-1}(m_f-m_0).
	\end{eqnarray*}
These usually don't match the boundary conditions $p(0)=p_0, p(t_f)=p_f$. When both the initial and terminal states are stationary, in which case $m_0 = p_0, m_f = p_f$, the discontinuity gaps at $t=0,\,t_f$  go to zero as $t_f$ goes to infinity.

Comparing \eqref{eq:workminnon} and \eqref{eq:W2gaussian} we again conclude the relation
	\begin{equation}\label{eq:OMTrelation}
		\mmW_{\rm min} = \Delta \mmF + \frac{1}{t_f} W_2(\rho_0,\rho_f) ^2.
	\end{equation}
Here we have employed the property that entropy is invariant with respect to translation. 	
Moreover, the resulting density flow $\rho_t$ is the (scaled) displacement interpolation between $\rho_0$ and $\rho_f$.

\section{Relaxation}\label{sec:relax}
In this section, we consider a modified version of Problem~\ref{pro:cov}. We specify a terminal value for the potential by fixing $Q_f$ while we relax the terminal constraint $\Sigma(t_f)=\Sigma_f$, which is commonly set to be $\frac{\sigma^2}{2}Q_f^{-1}$. This value for the covariance will be then attained asymptotically since, even if we do not specify the terminal distribution, it will converge to the Boltzmann distribution due to fluctuation-dissipation effects. Therefore, if our goal is to simply minimize the work, there is no need to insist on setting $\Sigma(t_f)=\Sigma_f$. More precisely, we address the following.
\begin{problem}\label{pro:covnew}
Find a function $Q(\cdot)$ from $Q(0)=Q_0$ to $Q(t_f)=Q_f$ over time $[0,t_f]$ that minimizes the average work \eqref{eq:workdisc} subject to constraint \eqref{eq:lyap} as well as the boundary condition $\Sigma(0)=\Sigma_0$.
\end{problem}

There are several possible approaches to solve the above problem. One of them is applying standard calculus of variations, just like what we did in the proof of Problem \ref{pro:cov}. Here, we adopt an alternative idea which solves the problem in two steps. We first find the solution for a given terminal value $\Sigma(t_f)=\Sigma_f$ and then minimize the cost function over all possible $\Sigma_f \ge 0$. Evidently, the first step is equivalent to solving Problem \ref{pro:cov}. The optimal cost is given by \eqref{eq:WFnon}, which is
	\[
		\mmF(\rho_f;\mH_f) - \mmF(\rho_0;\mH_0) +\frac{1}{t_f} W_2(\rho_0,\rho_f) ^2,
	\]
where $\rho_0, \rho_f$ are the zero-mean Gaussian distributions with covariances $\Sigma_0, \Sigma_{f}$. The free energy is
	\[
		\mmF(\rho_f;\mH_f) = \frac{1}{2}\tr(Q_f\Sigma_f)-\frac{\sigma^2}{4}\log\det\Sigma_f+{\rm constant}.
	\]
By \eqref{eq:OMTSDP}, the distance between two Gaussian distributions is given by the solution of the SDP
	\[
		\min_S~ \left\{\tr(\Sigma_0+\Sigma_f-2S) ~\mid~
		\left[\begin{matrix}
		\Sigma_0 & S \\ S' & \Sigma_f
		\end{matrix}\right]\ge 0\right\}.
	\]
Plugging them into the cost yields the convex optimization formulation of Problem \ref{pro:covnew}
	\begin{subequations}\label{eq:relax}
	\begin{eqnarray}
	&&\min_{S,\Sigma_f}~ \frac{1}{2}\tr(Q_f\Sigma_f)-\frac{\sigma^2}{4}\log\det\Sigma_f+\frac{1}{t_f}\tr(\Sigma_f-2S),
	\\
	&&\mbox{subject to} \quad \left[\begin{matrix}
		\Sigma_0 & S \\ S' & \Sigma_f
		\end{matrix}\right]\ge 0.
	\end{eqnarray}
	\end{subequations}
After solving \eqref{eq:relax}, we can obtain the solution of Problem \ref{pro:covnew} via that of Problem \ref{pro:cov} with the optimal $\Sigma_f$ as a boundary condition.
\begin{thm}\label{thm:solrelax}
The convex optimization problem \eqref{eq:relax} has a unique minimizer at
	\begin{subequations}
	\begin{eqnarray}
		\Sigma_f &=&  \frac{\sigma^2}{4}(\frac12Q_f+\frac{I}{t_f}+X)^{-1}
		\\
		S &=&  \Sigma_0^{1/2}(\Sigma_0^{1/2}\Sigma_f\Sigma_0^{1/2})^{1/2}\Sigma_0^{-1/2}
	\end{eqnarray}
	\end{subequations}
with
	\[
		X = \frac{2}{\sigma^2 t_f^2}\Sigma_0-\frac{2}{\sigma t_f} \Sigma_0^{1/2}\left(
		\Sigma_0^{-1/2}(\frac12Q_f+\frac{I}{t_f})\Sigma_0^{-1/2}+\frac{I}{\sigma^2t_f^2}\right)^{1/2}\Sigma_0^{1/2}.
	\]
\end{thm}
\begin{proof}
First we construct a Lagrangian
	\[
		\cL (S, \Sigma_f, \Psi) = \frac{1}{2}\tr(Q_f\Sigma_f)-\frac{\sigma^2}{4}\log\det\Sigma_f+\frac{1}{t_f}\tr(\Sigma_f-2S)+
		\tr\left(\left[\begin{matrix}
		\Sigma_0 & S \\ S' & \Sigma_f
		\end{matrix}\right]\Psi\right)
	\]
with Lagrange multiplier
	\[
		\Psi = \left[\begin{matrix} \Psi_{11} & \Psi_{12} \\ \Psi_{12}' & \Psi_{22}
		\end{matrix}\right] \le 0.
	\]
Minimizing $\cL$ over $S$ leads to the constraint
	\begin{equation}\nonumber
		\Psi_{12} = I/t_f,
	\end{equation}
and over $\Sigma_f$ yields as minimizer
	\begin{equation}\label{eq:Sigmaf}
		\Sigma_f = \frac{\sigma^2}{4} \left(\frac{1}{2}Q_f+\frac{I}{t_f}+\Psi_{22}\right)^{-1}.
	\end{equation}
Therefore, we obtain the dual problem
	\begin{subequations}\label{eq:dual}
	\begin{eqnarray}
		&&\max_{\Psi}~ \tr(\Sigma_0 \Psi_{11})+\frac{\sigma^2}{4} \log\det(\frac{1}{2}Q_f+\frac{I}{t_f}+\Psi_{22})
		\\&& \mbox{subject to} \quad \left[\begin{matrix}
		\Psi_{11} & \frac{I}{t_f} \\ \frac{I}{t_f} & \Psi_{22}
		\end{matrix}\right]\le 0.
	\end{eqnarray}
	\end{subequations}
In the above, for fixed $\Psi_{22}$, the minimizer over $\Psi_{11}$ is clearly $\Psi_{11} = \frac{1}{t_f^2} \Psi_{22}^{-1}$. Thus, \eqref{eq:dual} is equivalent to the convex optimization problem
	\begin{equation}\label{eq:Psi22}
		\max_{\Psi_{22}\le 0}~ \tr(\frac{1}{t_f^2}\Sigma_0\Psi_{22}^{-1})+\frac{\sigma^2}{4}\log\det (\frac{1}{2}Q_f+\frac{I}{t_f}
		+\Psi_{22}).
	\end{equation}
Its first order optimality condition is
	\[
		-\frac{1}{t_f^2}\Psi_{22}^{-1}\Sigma_0\Psi_{22}^{-1}+\frac{\sigma^2}{4}(\frac{1}{2}Q_f+\frac{I}{t_f}+\Psi_{22})^{-1}=0,
	\]
or equivalently
	\[
		-\frac{\sigma^2}{4}t_f^2\Psi_{22}\Sigma_0^{-1}\Psi_{22}+\Psi_{22}+\frac{1}{2}Q_f+\frac{I}{t_f}=0.
	\]
Let $X=\Sigma_0^{-1/2}\Psi_{22}\Sigma_0^{-1/2}$, then
	\[
		\frac{\sigma^2}{4}X^2 -X = Y:= \Sigma_0^{-1/2}(\frac{1}{2}Q_f+\frac{I}{t_f})\Sigma_0^{-1/2}.
	\]
It follows that
	\[
		(\frac{\sigma}{2}X-\frac{I}{\sigma t_f})^2 = Y+\frac{I}{\sigma^2 t_f^2}.
	\]
If we pick the solution
	\begin{equation}\label{eq:X}
		\frac{\sigma}{2}X =\frac{I}{\sigma t_f} -(Y+\frac{I}{\sigma^2 t_f^2})^{1/2}\le 0,
	\end{equation}
then $\Psi_{22}=\Sigma_0^{1/2}X\Sigma_0^{1/2}$ satisfies the constraint $\Psi_{22}\le 0$. Thus, in view of the strong convexity of \eqref{eq:Psi22}, we
conclude that $\Psi_{22}=\Sigma_0^{1/2}X\Sigma_0^{1/2}$ with $X$ in \eqref{eq:X} is the unique solution to \eqref{eq:Psi22}. The optimal $S, \Sigma_f$ follow from \eqref{eq:S} and \eqref{eq:Sigmaf}. This completes the proof.
\end{proof}

We establish similar results when the centers of the potentials are nonzero. Denote by $p_0$, $p_f$ the centers of the initial and target potentials. We seek an optimal control for the following problem.
\begin{problem}\label{pro:covnewnonzero}
Find $Q(\cdot)$ from $Q(0)=Q_0$ to $Q(t_f)=Q_f$ and $p(\cdot)$ from $p(0)=p_0$ to $p(t_f)=p_f$ over time $[0,t_f]$ that minimize the total work subject to constraint \eqref{eq:Qp} as well as boundary condition $\Sigma(0)=\Sigma_0, m(0) = m_0$.
\end{problem}

The idea is the same as in the zero-mean case. Straight forward calculation gives
	\[
		\mmF(\rho_f;\mH_f) = \frac{1}{2}\tr(Q_f\Sigma_f)-\frac{\sigma^2}{4}\log\det\Sigma_f+\frac12(m_f-p_f)'Q_f(m_f-p_f).
	\]
This together with \eqref{eq:OMTSDP} and \eqref{eq:OMTrelation} points to the convex optimization formulation
	\begin{subequations}\label{eq:relaxnonzero}
	\begin{eqnarray}
	&&\min_{S,\Sigma_f,m_f}~ \frac{1}{2}\tr(Q_f\Sigma_f)-\frac{\sigma^2}{4}\log\det\Sigma_f+\frac{1}{t_f}\tr(\Sigma_f-2S),
	\\
	&&\hspace{1.3cm}+\frac12(m_f-p_f)'Q_f(m_f-p_f)+\frac{1}{t_f}\|m_0-m_f\|^2 \nonumber
	\\
	&&\mbox{subject to} \quad \left[\begin{matrix}
		\Sigma_0 & S \\ S' & \Sigma_f
		\end{matrix}\right]\ge 0.
	\end{eqnarray}
	\end{subequations}
We note that in \eqref{eq:relaxnonzero} the minimization over $S,\Sigma_f$ and that over $m_f$ are decoupled. Thus, we have two independent optimization problems. The solution to the former is given by Theorem \ref{thm:solrelax} and the solution to the latter is given in closed-form as
	\begin{equation}\label{eq:mf}
		m_f = (Q_f+\frac{2}{t_f})^{-1}(Q_f p_f + \frac{2}{t_f}m_0).
	\end{equation}
Next we use the above results to recover two scalar cases that have been solved in Schmiedl and Seifert \cite{SchSei07}. 
\subsection{Case study I: Moving laser trap}
Suppose the strength $Q$ of the potential $\mH$ is fixed to be $Q\equiv 1$. Our goal is to choose proper function $p(\cdot)$ from $p(0) = p_0=0$ to $p(t_f) = p_f$ such that the work is minimized. The initial state is assumed to be at equilibrium, i.e., $m_0 = 0, \Sigma_0 = \sigma^2/2$.

When $Q_0=Q_f$ and $\Sigma_0=\frac{\sigma^2}{2}\Sigma_0^{-1}$, it can be easily seen from Theorem \ref{thm:solrelax} that the optimal strategy is $Q(t) \equiv Q_0$. Thus the assumption $Q\equiv 1$ is consistent with the optimality. Regarding the centers $p(t)$, from \eqref{eq:mf} we obtain
	\[
		m_f = \frac{t_f p_f}{2+t_f}.
	\]
Substituting it back to Theorem \ref{thm:optQp} we obtain that the optimal $p(t)$ is
	\[
		p(t) = \frac{t}{t_f}m_f +\frac{1}{t_f}m_f =\frac{(1+t) p_f}{2+t_f},
	\]
and the corresponding work is
	\[
		\mmW_{\rm min} = \frac{1}{t_f}m_f^2+\frac12(m_f-p_f)^2=\frac{p_f^2}{2+t_f}.
	\]
	
\subsection{Case study II: Time-dependent strength of the trap}
The means $p_0, p_f, m_0$ are set to be zero. The task is to drive the linear system \eqref{eq:OU} from an initial ``strength'' $Q_0$ to a terminal $Q_f$. The initial state is assumed to be zero stationary, namely, $\Sigma_0=\frac{\sigma^2}{2}Q_0^{-1}$. Applying Theorem \ref{thm:solrelax} we obtain
	\[
		\Sigma_f = \frac{\sigma^2}{2}\left[\sqrt{Q_fQ_0+2Q_0/t_f+1/t_f^2}-1/t_f\right]^{-2}.
	\]
This together with Theorem \ref{thm:opt} points to the optimal solution
	\begin{eqnarray*}
	Q(t) &=& \frac{Q_0-\Lambda(1+\Lambda t)}{(1+\Lambda t)^2}
	\\ \Sigma(t) &=& \Sigma_0(1+\Lambda t)^2
	\end{eqnarray*}
with
	\begin{eqnarray*}
		\Lambda &=& \frac{1}{t_f} (-1 + \sqrt{\Sigma_f/\Sigma_0})
		\\&=& \frac{\sqrt{Q_fQ_0t_f^2+2Q_0t_f+1}-1-Q_ft_f}{(2+Q_ft_f)t_f}.
	\end{eqnarray*}
	
\section{Connection to JKO gradient flow}\label{sec:JKO}
The result \eqref{eq:WFnon} is closely related to the celebrated Jordan-Kinderlehrer-Otto (JKO) flow \cite{JorKinOtt98}. In fact, we obtain an alternative proof of Theorem \ref{thm:WFnon} for general marginal distributions based on the results in \cite{JorKinOtt98}.

 The JKO scheme gives that the Fokker-Planck equation as the gradient flow of the free energy with respect to the Wasserstein metric $W_2$. Indeed, according to \cite{JorKinOtt98}, the Fokker-Planck equation
	\begin{equation}\label{eq:fokker}
	\frac{\partial \rho}{\partial t} - \nabla\cdot(\nabla \mH(x)\rho)-\frac{\sigma^2}{2}\Delta \rho = 0
	\end{equation}
can be viewed as the gradient flow of the free energy
	\begin{equation}
	\mmF(\rho;\mH) = \mmH(\rho)-TS(\rho) = \int \mH(x)\rho(x) +\frac{\sigma^2}{2}\int \rho\log\rho
	\end{equation}
with respect to the Wasserstein metric $W_2$ on the manifold of probability densities. More specifically, discretizing the above in the time domain, we obtain the celebrated JKO scheme. This amounts to the fact that $\rho^{k+1}(x):= \rho((k+1)h,x)$, where $\rho$ is the solution to \eqref{eq:fokker} and $h$ is the step size, minimizes
	\begin{equation}\label{eq:JKO}
	\frac{1}{2h} W_2(\rho,\rho^k)^2 + \mmF(\rho;\mH)-\mmF(\rho^k;\mH)
	\end{equation}
over $\rho$ as $h$ goes to $0$. This is akin to our result \eqref{eq:WFnon}. Next we discuss the connection between the two.

Since $\rho^{k+1}$ minimizes \eqref{eq:JKO}, we have
	\begin{equation}\label{eq:approx}
	\frac{1}{2h} W_2(\rho^{k+1},\rho^k)^2 + \mmF(\rho^{k+1};\mH)-\mmF(\rho^k;\mH) = -\frac{1}{2h} W_2(\rho^{k+1},\rho^k)^2 + o(h).
	\end{equation}
This is the Wasserstein counterpart of
	\[
		\frac{1}{2h} \|x-x_0\|^2 +f(x)-f(x_0) = -\frac{1}{2h} \|x-x_0\|^2 +o(h)
	\]
when $x$ minimizes the left-hand side (LHS), which follows from the approximation
	\[
		{\rm LHS} \approx \frac{1}{2h} \|x-x_0\|^2 +\nabla f(x_0)\cdot(x-x_0).
	\]
Rearranging \eqref{eq:approx} leads to
	\begin{equation}
	\frac{1}{h} W_2(\rho^{k+1},\rho^k)^2=  \mmF(\rho^k;\mH)-\mmF(\rho^{k+1};\mH) +o(h).
	\end{equation}
Now summing up the above we obtain
	\begin{equation}
	\frac{1}{h} \sum_{k=0}^{N-1} W_2(\rho^{k+1},\rho^k)^2\approx  \mmF(\rho_0)-\mmF(\rho_f)
	\end{equation}
where $N=t_f/h$ is the number of steps. Applying both Cauchy-Schwarz and the triangular inequality yields
	\begin{eqnarray}
	\nonumber
	\mmF(\rho_0;\mH)-\mmF(\rho_f;\mH)  &\approx&  \frac{1}{t_f}{N} \sum_{k=0}^{N-1} W_2(\rho^{k+1},\rho^k)^2
	\\
	&\approx&\frac{1}{t_f} \left(\sum_{k=0}^{N-1} W_2(\rho^{k+1},\rho^k)\right)^2\nonumber
	\\
	&\ge& \frac{1}{t_f} W_2(\rho_0,\rho_f)^2.\label{eq:Cauchy}
	\end{eqnarray}
Finally, by letting $h$ goes to $0$ we establish
	\[
		\frac{1}{t_f} W_2(\rho_0,\rho_f)^2 +\mmF(\rho_f;\mH)-\mmF(\rho_0;\mH) \le 0,
	\]
which is a special case of \eqref{eq:WFnon} when $\mmW=0$. Indeed, when the potential $\mH$ is time-invariant, there is no work being done.

When $\mH$ is time-varying, the analysis is similar. The Fokker-Planck equation is
	\begin{equation}\label{eq:fokkervarying}
	\frac{\partial \rho}{\partial t} - \nabla\cdot(\nabla \mH_t(x)\rho)-\frac{\sigma^2}{2}\Delta \rho = 0
	\end{equation}
and the approximation \eqref{eq:approx} becomes
	\begin{align*}
	&\frac{1}{2h} W_2(\rho^{k+1},\rho^k)^2 + \mmF(\rho^{k+1};\mH^{k+1})-\mmF(\rho^k;\mH^k)\\&= -\frac{1}{2h} W_2(\rho^{k+1},\rho^k)^2 +
	h\int\frac{\partial \mH_t}{\partial t} \rho^k dx+ o(h),
	\end{align*}
where $\mH^k:=\mH_{kh}$.
Following the same steps as in the time-invariant case, we conclude
	\begin{equation}\label{eq:WF}
		\frac{1}{t_f} W_2(\rho_0,\rho_f)^2 +\mmF(\rho_f;\mH_f)-\mmF(\rho_0;\mH_0) \le \int_0^{t_f}\int\frac{\partial \mH_t}{\partial t} \rho dxdt =
		\mmW.
	\end{equation}
	
Thus, the amount of work $\mmW$ needed is always lower-bounded by the change of free energy plus the optimal transport cost between the marginal state distributions $\rho_0, \rho_f$. Moreover, in view of \eqref{eq:Cauchy}, the equality holds when the density flow $\rho_t$ is the displacement interpolation between $\rho_0$ and $\rho_f$. The optimal Hamiltonian is
	\[
		\mH_t = -\phi_t-\frac{\sigma^2}{2}\log \rho,
	\]
where $\phi_t$ is the unique solution (under the constraint that $\int \phi_t = 0$) to
	\[
		\frac{\partial \rho_t}{\partial t} +\nabla\cdot(\rho_t \nabla \phi_t) = 0,
	\]
which, by standard optimal mass transport theory \cite{BenBre00}, always exists.

Therefore, we conclude that {\em Theorem \ref{thm:WFnon} holds for general marginal distributions $\rho_0, \rho_f$} providing we are free to change the Hamiltonian in whatever way we like.
The above analysis also reveals that, for general time-varying Hamiltonian, we can calculate the work using the solution $\rho_t$ to \eqref{eq:fokkervarying}.

\begin{thm}\label{thm:diss2}
	\begin{subequations}\label{eq:Wdiss}
	\begin{eqnarray}
	\mmW-\mmF(\rho_f;\mH_f)+\mmF(\rho_0;\mH_0) &=& \int_0^{t_f} \int \rho_t(x) \|\nabla\phi_t(x)\|^2 dx dt
	\\&&
	\frac{\partial \rho_t}{\partial t} +\nabla\cdot(\rho_t \nabla \phi_t) = 0.
	\end{eqnarray}
	\end{subequations}
\end{thm}
We note that this is a generalization to Theorem \ref{thm:diss} to general marginal distributions.

\begin{remark}
Finally, we note the counterpart of the above result for gradient dynamics
	\[
		\dot x(t) = -\nabla f(t,x(t))
	\]
in a Euclidean space. In general,
	\[
		\frac{1}{t_f} \|x(t_f)-x(0)\|^2 + f(t_f,x(t_f))-f(0,x(0)) \le \int_0^{t_f} \frac{\partial f}{\partial t} (t,x(t))dt.
	\]
In the special case when $f$ is independent of time,
	\[
		\frac{1}{t_f} \|x(t_f)-x(0)\|^2 + f(x(t_f))-f(x(0)) \le 0.
\]
\end{remark}

\section{The Jarzynski equality}\label{sec:jarzynski}

Different from the above results, the Jarzynski equality \cite{Jar97a,Jar97b} provides an alternative way to compare work and free energy. It reads,
	\begin{equation}\label{eq:Jarzynski}
	\mE \{\exp (-\beta \mW(Q,x))\} = \exp (-\beta\Delta \mmF),
	\end{equation}
under the assumption that the initial state is at equilibrium. Here $\Delta \mmF$ is the difference of free energy at equilibrium, viz.,
	\[
		\Delta \mmF = \mmF(\rho_B;\mH_f) - \mmF(\rho_B;\mH_f).
	\]
Note that the Jarzynski equality implies $\mmW=\mE\{\mW\} \ge \Delta \mmF$ by Jensen's inequality. Indeed,
	\[
		\exp(-\beta \mmW) \le \mE \{\exp (-\beta \mW(Q,x))\} = \exp (-\beta\Delta \mmF).
	\]
The relation $\mmW \ge \Delta \mmF$ then follows from the monotonicity of the exponential function.


We next recall a simple derivation of the Jarzynski equation. Let $\rho(t,\cdot)$ be the density of $x(t)$, and let
\[
g(t,y)=\mE\{\exp(-\beta\int_0^t \frac{\partial \mH_s}{\partial s} ds)\,\mid \, x(t)=y\}.
\]
We have that
	\begin{equation}\label{eq:g}
		\frac{\partial g}{\partial t}+(-\nabla \mH_t-\sigma^2\nabla\log \rho)\cdot\nabla g -\frac{\sigma^2}{2}\Delta g
		+\beta\frac{\partial \mH_t}{\partial t}g = 0,~~g(0,\cdot) \equiv 1.
	\end{equation}
To see this, first rewrite \eqref{eq:OU} in the reverse direction utilizing Doob's $h$-transform,
	\begin{equation}
		dx(t) = (-\nabla \mH_t -\sigma^2 \nabla\log \rho)dt + \sigma dw_-.
	\end{equation}
This is quite standard. Here, $dw_-$ is a reverse Wiener process, whose current increment is independent with the future. We deduce that
	\begin{align*}
		g(t,x) &= \mE \left\{\exp(-\beta \int_0^t \frac{\partial \mH_s}{\partial s}(x(s))ds)~\mid~ x(t)=x\right\}
		\\
		&\hspace*{-.3in}\approx \mE \left\{ \mE\left[\exp(-\beta \int_0^{t-dt} \frac{\partial \mH_s}{\partial s}(x(s))ds)
		~\mid~x(t-dt) =x+dx\right]\right.\\
		&\left.\phantom{\int_0^{t-dt}}
		\times(1-\beta\frac{\partial \mH_t}{\partial t}(x)dt)~\mid x(t)=x\right\}
		\\
		&= \mE \left\{ g(t-dt,x+dx)(1-\beta\frac{\partial \mH_t}{\partial t}(x)dt)~\mid x(t)=x\right\}
		\\
		&\approx \mE \left\{ g(t,x)-\frac{\partial g}{\partial t} dt-\nabla g\cdot
		(-\nabla \mH_t -\sigma^2 \nabla\log \rho)dt\right.\\
		&\left.\phantom{\frac{\partial g}{\partial t}}+\sigma \nabla g\cdot dw_-
		+\frac{\sigma^2}{2} \Delta gdt
		-\beta g\frac{\partial \mH_t}{\partial t}(x)dt~\mid x(t)=x\right\}
		\\
		&= g(t,x)-\frac{\partial g}{\partial t} dt-\nabla g\cdot
		(-\nabla \mH_t -\sigma^2 \nabla\log \rho)dt\\&\phantom{\frac{\partial g}{\partial t}}+\frac{\sigma^2}{2} \Delta gdt
		-\beta g\frac{\partial \mH_t}{\partial t}(x)dt,
	\end{align*}
from which \eqref{eq:g} follows. Combining \eqref{eq:g} and the Fokker-Planck equation
	\[
		\frac{\partial \rho} {\partial t} +\nabla\cdot(-\nabla \mH_t \rho)-\frac{\sigma^2}{2}\Delta \rho = 0,
	\]
we establish that $f:=g\rho$ satisfies
	\begin{eqnarray*}
	\frac{\partial f}{\partial t} &=& \frac{\partial g}{\partial t} \rho + g\frac{\partial \rho}{\partial t}
	\\
	&=& \nabla\cdot(\nabla \mH_t f)+\frac{\sigma^2}{2} \Delta f - \beta \frac{\partial \mH_t}{\partial t} f.
	\end{eqnarray*}

Finally, we claim that
	\[
		f(t,x) = \frac{1}{Z_0}\exp (-\beta \mH_t(x))
	\]
with $Z_0 = \int \exp (-\beta \mH_0(x))dx$.
To see this, we just need to notice
	\[
		\frac{\partial f}{\partial t} = - \beta \frac{\partial \mH_t}{\partial t} f
	\]
as well as
	\[
		\nabla\cdot(\nabla \mH_t f)+\frac{\sigma^2}{2} \Delta f =0.
	\]
Clearly the boundary condition $g(0,\cdot) \equiv 1$ also holds as
	\[
		f(0,x) = \frac{1}{Z_0}\exp (-\beta \mH_0(x)) = \rho_0.
	\]
	
\section{Numerical examples}\label{sec:eg}
Consider a laser trap in one dimensional space whose strength and center location can both vary over time. The initial Hamiltonian $\mH_0$ and terminal Hamiltonian $\mH_f$ are set to be quadratic with parameters
	\begin{eqnarray*}
	Q_0 &=& 1,~~p_0 = 0.3,
	\\
	Q_f  &=& 4,~~p_f = -1,
	\end{eqnarray*}
respectively. The initial state distribution is assumed to be stationary, that is, 
	\[
		\Sigma_0 = \frac{\sigma^2}{2}Q_0^{-1}=\frac{\sigma^2}{2},~~m_0=p_0=0.3.
	\]
We further choose $\sigma$ to be $1$ and $f$ to be $1$ to simplify the calculation. Our goal is to drive the state distribution to the stationary distribution corresponding to $\mH_1$, which is 
	\[
		\Sigma_1 = \frac{1}{2}Q_1^{-1} = \frac{1}{8},~~m_1 = p_1 = -1,
	\]
via adjusting the Hamiltonian $\mH_t$. The optimal strategy is given in Section \ref{sec:regulation} if we want to achieve this distribution at $t=1$. Plugging the above parameters into Theorem \ref{thm:opt}, we get the optimal strategy and density flow being 
	\[
		Q(t) = \frac{6-t}{(t-2)^2},
	\]
	\[
		p(t) = 0.3-1.3t-\frac{1.3(t-2)^2}{6-t},
	\]
and 
	\[
		\Sigma(t) = \frac{(t-2)^2}{8}
	\]
respectively for all $t\in (0, 1)$. The minimum work is $2.162$. The Hamiltonian jumps from $(Q(1)=5,\, p(1)=-1.26)$ to $(Q_1,\,p_1)$ at the terminal time point $t=1$, after which, both the Hamiltonian and state density remain time-invariant. Figure \ref{fig:rhointerp1} depicts the evolution of the probability density of the state. Several typical sample paths are plotted in Figure \ref{fig:path1}. Clearly the sample paths are consistent with the density flow.

\begin{figure}[h]
\centering
\includegraphics[width=0.5\textwidth]{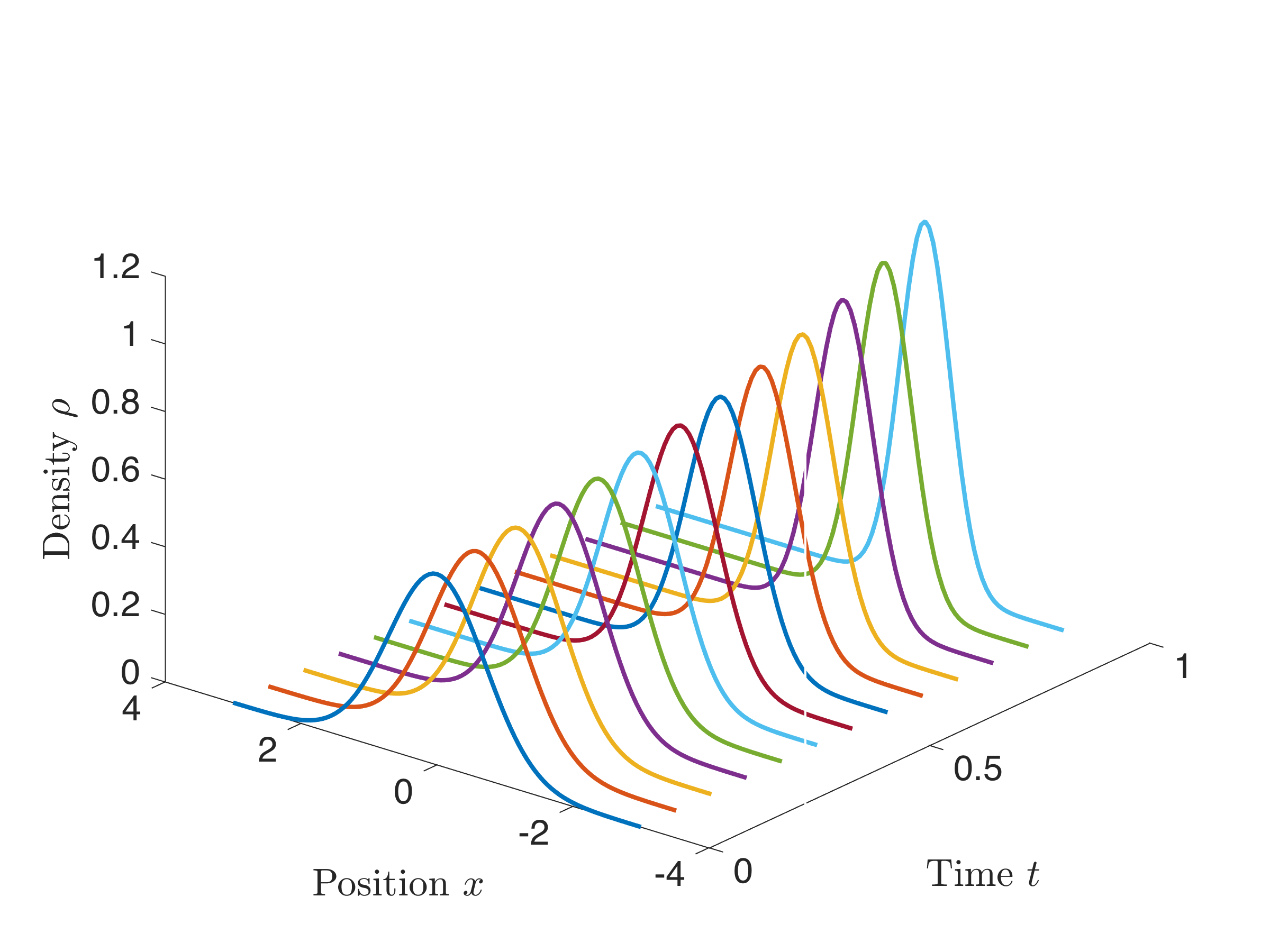}
\caption{Density evolution}
\label{fig:rhointerp1}
\end{figure}
\begin{figure}[h]
\centering
\includegraphics[width=0.5\textwidth]{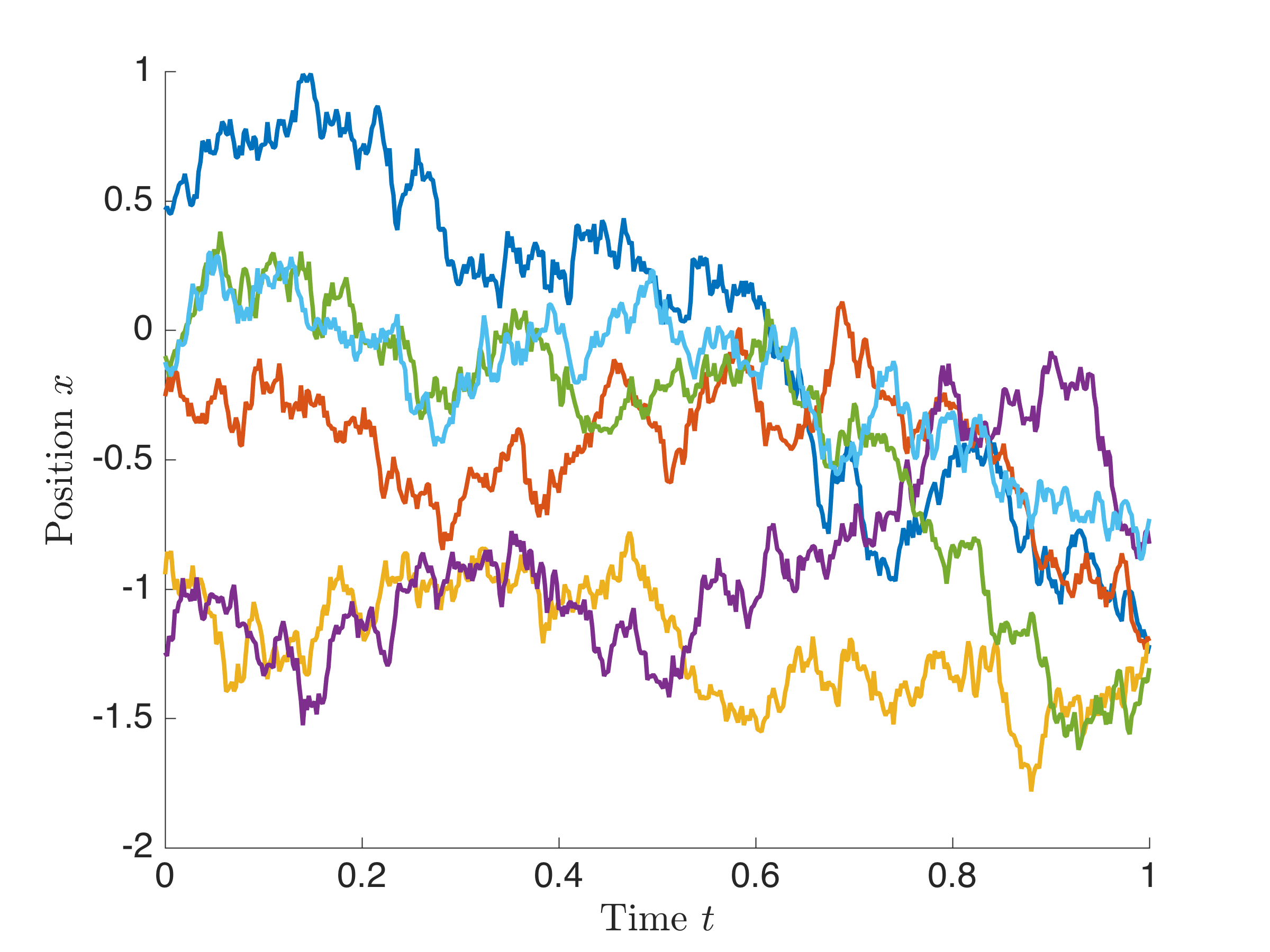}
\caption{Sample paths over $t\in [0, 1]$}
\label{fig:path1}
\end{figure}

Next we move to another scenario discussed in Section \ref{sec:relax} where the constraint on the terminal state distribution doesn't exist. Following the discussion in Section \ref{sec:relax}, we obtain the optimal terminal distribution at $t=1$ to be 
	\[
		\Sigma_1 = 0.3273, ~~m_1 = -0.35.
	\]
The corresponding optimal strategy and density flow can be again obtained using Theorem \ref{thm:opt}.
The minimum work is $0.9692$, which is less than $2.162$ in the previous setting. In this case, the terminal distribution is not stationary with respect to $\mH_1$ anymore, therefore, the state density will vary after the terminal time $t=1$. Eventually, due to fluctuation, the state density will converge to the stationary distribution $\cN(-1,\, 1/8)$. In Figure \ref{fig:rhointerp2}, we can see clearly that the evolution of the state distribution doesn't match the terminal condition $\cN(-1,\, 1/8)$. This can also be seen from the sample paths in Figure~\ref{fig:path2}. However, if we run the system long enough, then the state distribution will converge to the stationary one, as shown in Figure~\ref{fig:path3}, due to fluctuation.
\begin{figure}[h]
\centering
\includegraphics[width=0.5\textwidth]{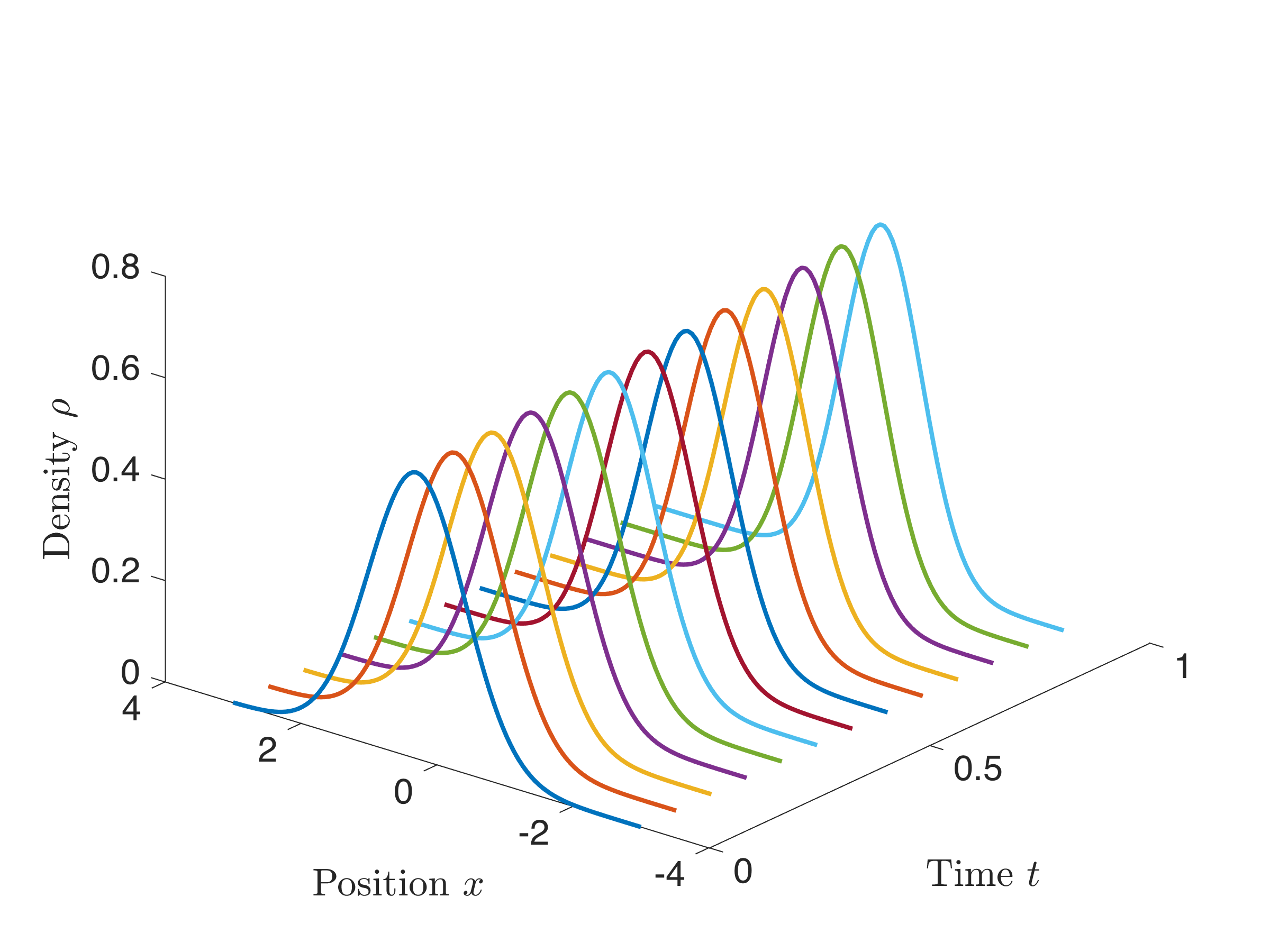}
\caption{Density evolution}
\label{fig:rhointerp2}
\end{figure}
\begin{figure}[h]
\centering
\includegraphics[width=0.5\textwidth]{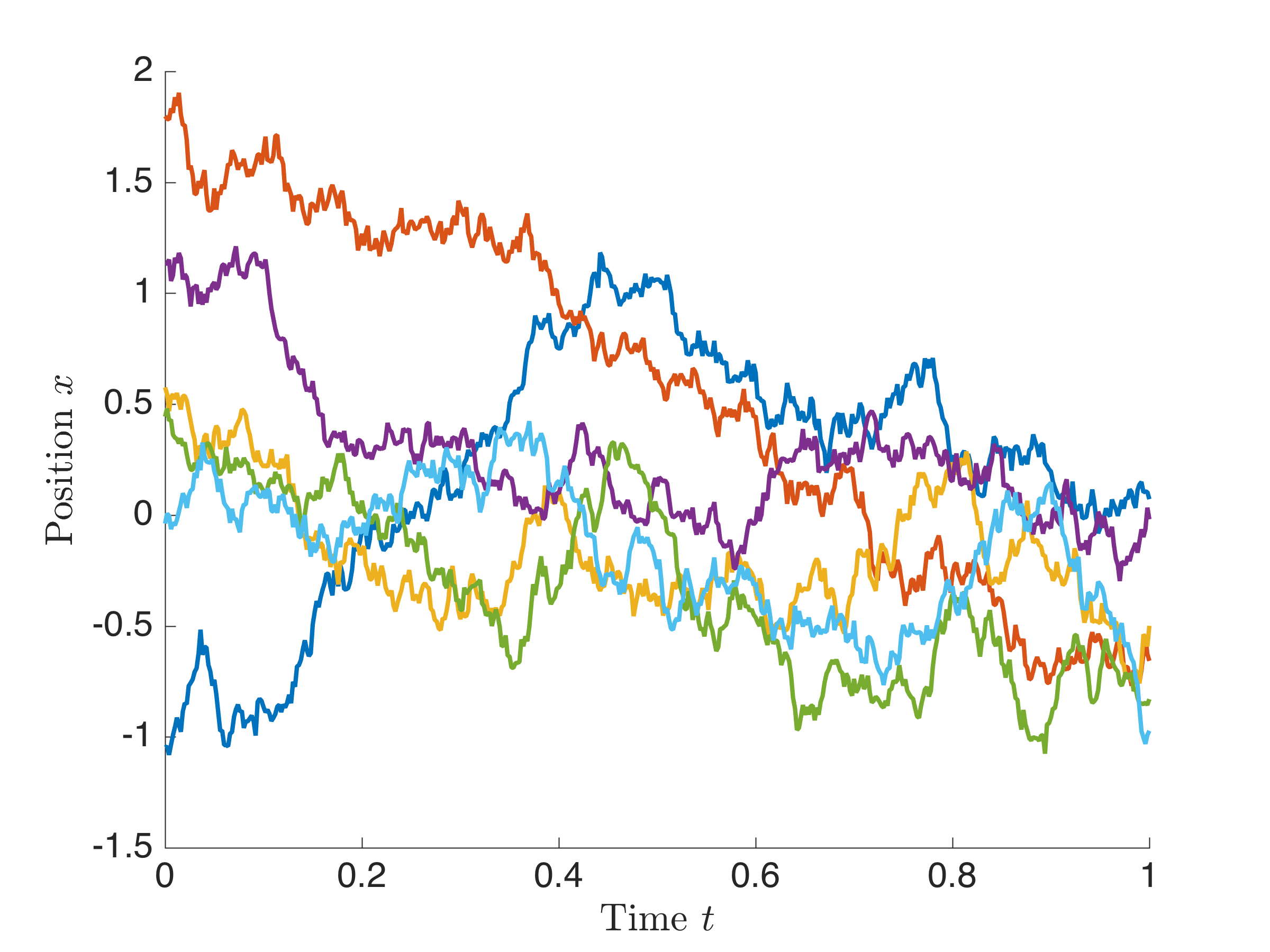}
\caption{Sample paths over $t\in [0, 1]$}
\label{fig:path2}
\end{figure}
\begin{figure}[h]
\centering
\includegraphics[width=0.5\textwidth]{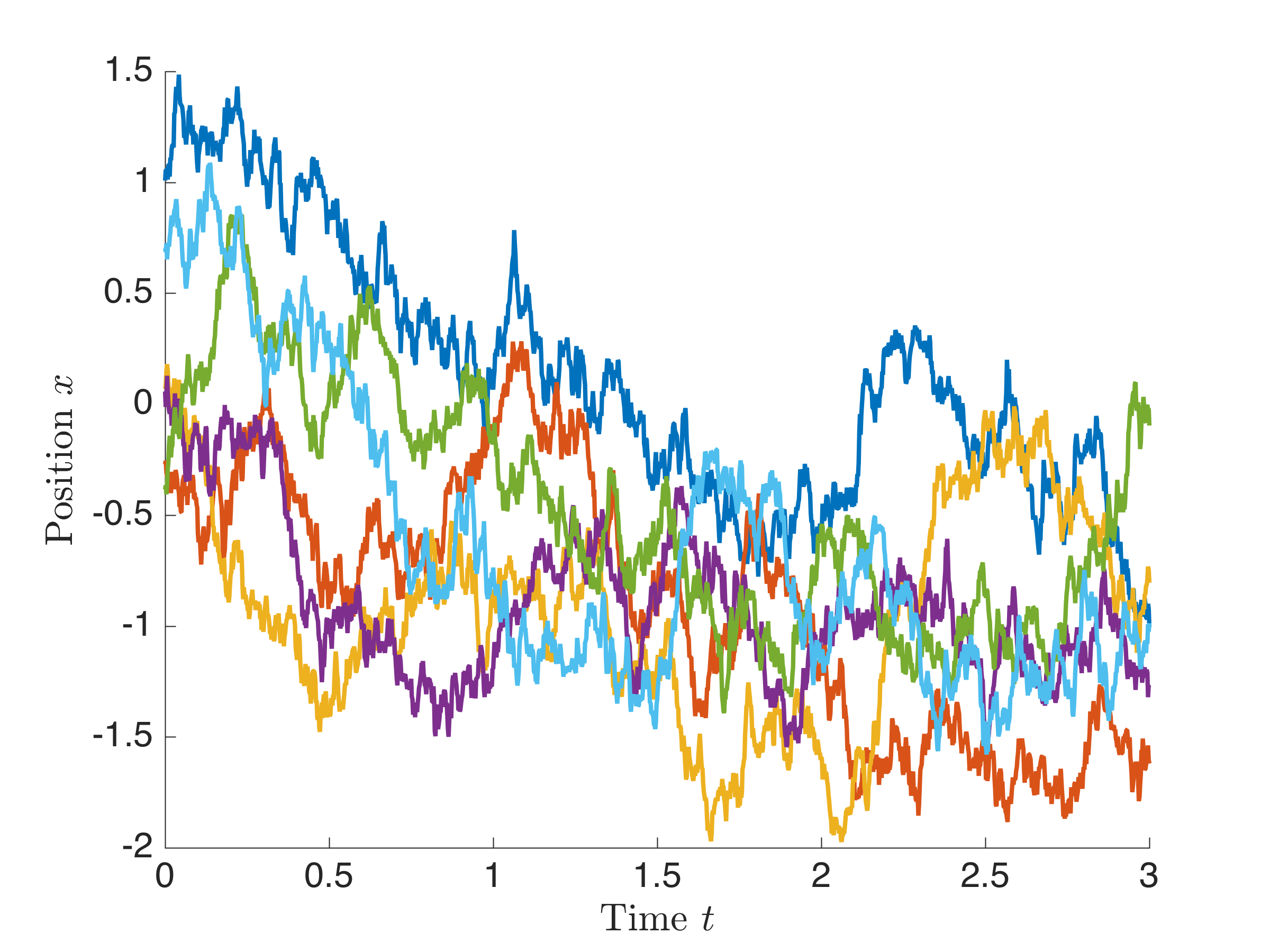}
\caption{Sample paths over $t\in [0, 3]$}
\label{fig:path3}
\end{figure}

\section{Conclusion}
We described the problem of controlling non-equilibrium thermodynamical systems with harmonic Hamiltonian, from a given initial state to a final state, by adjusting the parameter specifying the Hamiltonian in finite time. This led to some interesting connections to optimal mass transport theory, and gave a new twist to understanding the second law of thermodynamics building on the seminal results of Jarzynski and the recent advances in stochastic thermodynamics, e.g., see \cite{Sei12}.
We expect that the theory will lead to insights in the case of potentials with several wells and to connections with the theory of large deviations and information theory with optimal mass transport \cite{ParHorSag15,Leo12,Leo14}. We anticipate interesting connections to the celebrated Landauer principle \cite{Lan61}, which provides a fundamental lower bound of the energy consumption to erase one bit of information. In recent years, many experiments have been performed aiming to achieve the bound $k_B T\ln 2$ \cite{BerAraPetCil12, TalBhaSal17,talukdar2016energetics}.
This bound, however, theoretically can only be achieved through reversible processes. The constraint to erase a bit in finite time, unavoidably, introduces a gap. Our aim is to gain insights into such a gap using optimal transport theory and stochastic control.

\section*{Acknowledgements}
This project was supported by AFOSR grants (FA9550-15-1-0045 and FA9550-17-1-0435), ARO grant (W911NF-17-1-049), grants from the National Center for Research Resources (P41-RR-013218) and the National Institute of Biomedical Imaging and Bioengineering (P41-EB-015902), NCI grant (1U24CA18092401A1), NIA grant (R01 AG053991), and a grant from the Breast Cancer Research Foundation.



\bibliographystyle{IEEEtran}
\bibliography{./refs}

\end{document}